\documentclass{easychair}
\usepackage{a4wide}
\usepackage{amssymb}
\usepackage{tabularx}
\usepackage{multirow}
\usepackage{tikz}
\usetikzlibrary{positioning}
\usepackage{algorithm}
\usepackage[noend]{algpseudocode}
\usepackage{stmaryrd}

\title{On the Strong Metric Dimension of directed co-graphs}
\author{Yannick Schmitz \and Egon Wanke}
\institute{Heinrich-Heine-Universit\"at D\"usseldorf, Germany\\ \email{yannick.schmitz@hhu.de, egon.wanke@hhu.de}}
\newtheorem{theorem}{Theorem}
\newtheorem{definition}{Definition}

\newtheorem{lemma}{Lemma}

\newtheorem{corollary}{Corollary}

\authorrunning{Schmitz, Wanke}
\titlerunning{The strong metric dimension of directed co-graphs}

\newcommand{\DP}[3]{
\begin{center}
\begin{tabularx}{.75\textwidth}{lL}
\hline\hline
\multicolumn{2}{c}{\sc{#1}} \\
\hline
\em{Instance:}& #2\\
\em{Question:}& #3 \\ 
\hline\hline
\end{tabularx}
\end{center}
}

\newcommand{\End}{\State {\bf end}}

\newcommand{\GSR}{G_{\rm SR}}

\newcommand{\SMD}{\rm smd}

\newcommand{\GSRP}{G'_{\rm SR}}
\newcommand{\GSRC}{\overline{G_{\rm SR }}}
\newcommand{\MDT}{\,{\rm MDT}\,}
\newcommand{\MDF}{\,{\rm MDF}\,}
\newcommand{\MMDT}{\,{\rm MMDT}\,}

\newcommand{\UVWPATH}[3]{$#1\kern-3pt\shortrightarrow\kern-3pt#2\kern-3pt\shortrightarrow\kern-3pt#3$-path}
\newcommand{\union}{\hbox{\small$\cup$}\,}
\newcommand{\join}{\hbox{$\times$}\,}
\newcommand{\rjoin}{\raise1pt\hbox{\scriptsize$\gg$}\,}

\newcolumntype{L}{>{\raggedright\arraybackslash}X}
\newcolumntype{R}{>{\raggedleft\arraybackslash}X}
\newcolumntype{C}{>{\centering\arraybackslash}X}

\begin{document}

\maketitle

\begin{abstract}
Let $G$ be a strongly connected directed graph and $u,v,w\in V(G)$ be three vertices. Then $w$ {\em strongly resolves} $u$ to $v$ if there is a shortest $u$-$w$-path containing $v$ or a shortest $w$-$v$-path containing $u$. A set $R\subseteq V(G)$ of vertices is a {\em strong resolving set} for a directed graph $G$ if for every pair of vertices $u,v\in V(G)$ there is at least one vertex in $R$ that strongly resolves $u$ to $v$ and at least one vertex in $R$ that strongly resolves $v$ to $u$. The distances of the vertices of $G$ to and from the vertices of a strong resolving set $R$ uniquely define the connectivity structure of the graph. The {\em Strong Metric Dimension} of a directed graph $G$ is the size of a smallest strong resolving set for $G$. The decision problem {\sc Strong Metric Dimension} is the question whether $G$ has a strong resolving set of size at most $r$, for a given directed graph $G$ and a given number $r$. In this paper we study undirected and directed co-graphs and introduce linear time algorithms for {\sc Strong Metric Dimension}. These algorithms can also compute strong resolving sets for co-graphs in linear time.
\end{abstract}

%------------------------------------------------------------------------------
\section{Introduction}

The strong metric dimension is a modified version of the metric dimension. Both terms were first defined and analysed for undirected graphs. In order to embed the results worked out in this paper in the literature, we first need the precise definitions of the metric dimension for undirected and directed graphs in the normal and strong version. 

A vertex $w$ of an undirected graph $G$ {\em resolves} two vertices $u$ and $v$ if the distance between $w$ and $u$ is different from the distance between $w$ and $v$, that is, if $d(w,u) \not = d(w,v)$.
A set of vertices $R\subseteq V(G)$ is called a {\em resolving set} for $G$ if every pair of vertices of $G$ is resolved by at least one vertex of $R$.
The {\em metric dimension} of $G$, denoted by ${\rm md}(G)$, is the size of a smallest resolving set, which is also called a {\em metric basis} for $G$.

The {\em strong metric dimension} differs from the usual metric dimension only in the concept of resolving two vertices. A vertex $w$ {\em strongly resolves} $u$ and $v$ if there is either an undirected shortest path between $w$ and $u$ in $G$ containing $v$ or an undirected shortest path between $w$ and $v$ in $G$ containing $u$.
A set of vertices $R\subseteq V(G)$ is called a {\em strong resolving set} for $G$ if every pair of vertices of $G$ is strongly resolved by at least one vertex of $R$.
The {\em strong metric dimension} of $G$, denoted by ${\SMD}(G)$, is the size of a smallest strong resolving set, which is also called a {\em strong metric basis} for $G$.

Both models can be extended to directed graphs. This leads to the definition of the directed metric dimension and the strong directed metric dimension. In the directed metric dimension, the vertices $u$ and $v$ are resolved by vertex $w$ if the length of a shortest path from $w$ to $u$ and the length of a shortest path from $w$ to $v$ are different or the length of a shortest path from $u$ to $w$ and the length of a shortest path from $v$ to $w$ are different.

For the definition of the strong metric dimension for directed graphs, there are more case distinctions. Let us denote a directed shortest path from a vertex $x$ to a vertex $z$ that contains a vertex $y$ as a shortest \UVWPATH{x}{y}{z}. Two vertices $u$ and $v$ are strongly resolved from $u$ to $v$ by a vertex $w$ if there is a shortest \UVWPATH{w}{u}{v} or a shortest \UVWPATH{u}{v}{w}. They are strongly resolved if they are strongly resolved from $u$ to $v$ and strongly resolved from $v$ to $u$.

Due to the similarity of both definitions for the directed and undirected version, the decision problem {\sc Strong Metric Dimension} can be defined as follows for both cases:

\DP{Strong Metric Dimension}{A (directed or undirected) graph $G$ and a number $r$}{Is there a strong resolving set $R\subseteq V(G)$ for $G$ of size at most $r$?}

The metric dimension of undirected graphs has been introduced in the 1970s independently by Slater \cite{Sla75} and by Harary and Melter \cite{HM76}. The metric dimension finds applications in various areas, including network discovery and verification \cite{BEEHHMR05}, geographical routing protocols \cite{LA06}, combinatorial optimization \cite{ST04}, sensor networks \cite{HW12}, robot navigation \cite{KRR96} and chemistry \cite{CEJO00,Hay17}. The strong metric dimension was introduced by Seb{\"o} and Tannier \cite{ST04} who also showed, that a strong metric basis uniquely defines a graph. As both the metric and the strong metric dimension are NP-complete for general graphs \cite{KRR94, OP07}, special classes of graphs for which those problems can be solved efficiently are of interest. The directed versions of both problems are NP-complete too, as they are equivalent to their undirected version if each undirected edge is replaced by two directed edges.

There are several algorithms for computing a minimum resolving set in polynomial time for such classes, for example trees \cite{CEJO00,KRR96}, wheels \cite{HMPSCP05}, grid graphs \cite{MT84}, $k$-regular bipartite graphs \cite{BBSSS11}, amalgamation of cycles \cite{IBSS10}, cactus block graphs \cite{HEW16}, graphs with size-constrained biconnected components \cite{VHW19} and outerplanar graphs \cite{DPSL12}. The approxibility of the metric dimension has been studied for bounded degree, dense and general graphs in \cite{HSV12}. Upper and lower bounds on the metric dimension are considered in \cite{CGH08,CPZ00} for further classes of graphs.

The strong metric dimension is mainly studied on different products of graphs \cite{YKR11,KYR13,KRY14,RYKO14,KYR15,KYR16,KRY16}, but there are also efficient algorithms for other classes of graphs, for example trees \cite{ST04}, distance hereditary graphs \cite{MO11} and split graphs \cite{MOY17}. However, the strong metric dimension has not been extensively studied for directed graphs. It was only defined by Seb{\"{o}} and Tannier in \cite{ST04} in order to uniquely define the edges in a graph using the directed shortest path lengths to the vertices of a strong resolving set of vertices.

In this paper we present an algorithm that computes the strong metric dimension for undirected and directed co-graphs in linear time. This is one of the first approaches for the computation of the strong metric dimension of special directed graph classes. In Section 2 we first describe a linear time algorithm for the computation of the strong metric dimension for undirected co-graphs. This also illustrates the algorithmic idea for the linear-time algorithm presented in Section 3 for the computation of the strong metric dimension for directed co-graphs. 

\subsection{Preliminaries}
All graphs in this paper are finite, simple and if not stated otherwise, connected or strongly connected, respectively. For a (directed) graph $G=(V,E)$ with vertex set $V$ and edge set $E$, we write $V(G)$ for vertex set $V$ and $E(G)$ for edge set $E$ to reduce the number of variable names when using several graphs.

For a (strongly) connected graph $G$ and two vertices $u,v \in V(G)$ the distance $d_G(u,v)$ between $u$ and $v$ in $G$ is the length (number of edges) of a shortest path between $u$ and $v$ in $G$.
If it is clear from the context which graph is meant, we also write $d(u,v)$ instead of $d_G(u,v)$. 

\subsubsection{Undirected Graphs}
The {\em diameter} $d(G)$ of $G$ is the maximum distance between two vertices $u$ and $v$ of $G$.

Graph $G'$ is a {\em subgraph} of $G$, denoted by $G'\subseteq G$, if $V(G')\subseteq V(G)$ and $E(G') \subseteq \{\{u,v\}\ \vert\  u,v\in V(G'), \ \{u,v\}\in E(G)\}$.
Subgraph $G'$ is an {\em induced subgraph} of $G$, if $E(G') = \{\{u,v\}\ \vert\ u,v\in V(G'),\ \{u,v\}\in E(G)\}$.
We denote the subgraph of $G$ induced by the vertices of $V(G')$ by $G[V(G')]$ or $G[G']$.
For a vertex $u\in V(G)$ or a vertex set $U \subseteq V(G)$, let $G \setminus u = G[V(G) \setminus \{u\}]$ and $G \setminus U = G[V(G) \setminus U]$ be the subgraph of $G$ induces by $V(G) \setminus \{u\}$ and $V(G) \setminus U$, respectively.

For a vertex $u\in V(G)$, the {\em open neighbourhood} of $u$ is defined by $N_G(u) = \{v\ \vert\ v\in V(G),\ \{v,u\}\in E(G)\}$ and the {\em closed neighbourhood} of $u$ is defined by $N_G[u] = N_G(u)\cup \{u\}$.
If it is clear from the context which graph is meant, we also write $N(u)$ and $N[u]$ instead of $N_G(u)$ and $N_G[u]$. 

Two vertices $u,v \in V(G)$ are called {\em twins} if $N(u) = N(v)$.
Two twins $u$ and $v$ are called {\em true twins} if they are adjacent, that is, if $N[u] = N[v]$, otherwise they are called {\em false twins}.

The complement graph $\overline{G}$ of $G$ is defined by $V(\overline{G}) = V(G)$ and $E(\overline{G}) = \{\{u,v\}\ \vert\ u,v\in V(G),\ u \not=v,\ \{u,v\}\not\in E(G)\}$.  

The size of a minimum vertex cover, maximum independent set and maximum clique of an undirected graph $G$ is denoted by $\tau(G)$, $\alpha(G)$ and $\omega(G)$, respectively.

\subsubsection{Directed Graphs}
The {\em diameter} of $G$, denoted by $d(G)$, is the maximal distance from a vertex $u$ to a vertex $v$ in $G$.

Graph $G'$ is a {\em subgraph} of $G$, denoted by $G'\subseteq G$, if $V(G')\subseteq V(G)$ and $E(G') \subseteq \{(u,v)\ \vert\ u,v\in V(G'),\ (u,v) \in E(G)\}$.
Subgraph $G'$ is an {\em induced subgraph} of $G$, if $E(G') = \{(u,v)\ \vert\ u,v\in V(G'),\ (u,v) \in E(G)\}$.

For a vertex $u\in V(G)$, we denote by $N^-(u) = \{v\ \vert\ v\in V(G),\ (v,u)\in E(G)\}$ the set of {\em in-neighbours} of $u$ and by $N^+(u) = \{v\ \vert\ v\in V(G),\ (u,v)\in E(G)\}$ the set of {\em out-neighbours} of $u$.

For a vertex $u\in V(G)$ or a vertex set $U \subseteq V(G)$, let $G \setminus u = G[V(G) \setminus \{u\}]$ and $G \setminus U = G[V(G) \setminus U]$ be is the subgraph of $G$ induces by $V(G) \setminus \{u\}$ and $V(G) \setminus U$, respectively.

%------------------------------------------------------------------------------
\section{The strong metric dimension of undirected co-graphs}\label{sec:undirected}
In this section, we provide a linear time algorithm to compute the strong metric dimension of undirected co-graphs. May and Oellerman presented in \cite{MO11} an algorithm that computes the strong metric dimension for distance hereditary graphs in polynomial time.
Since co-graphs are distance hereditary, their algorithm computes the strong metric dimension also for co-graphs in polynomial time. However, the complicated data structures used by the algorithm of May and Oellerman are not necessary to compute the strong metric dimension for co-graphs, as our simple algorithm shows.
We introduce a different approach by analysing the co-tree, which we will later use to compute the strong metric dimension of directed co-graphs.

%
% Definition 1
%
\begin{definition}[\bf Undirected Co-Graphs]
\cite{CLB81}
\begin{itemize}
\item An undirected graph $G$ that consists of a single vertex $v$ is an undirected co-graph. The co-tree $T$ of $G$ consists of a single node $u$ associated with vertex $v$ of $G$. Node $u$ is the {\em root} of $T$.

\item If $G_1$ and $G_2$ are two undirected co-graphs, then the {\em union} of $G_1$ and $G_2$, denoted by $G_1 \union G_2$, is an undirected co-graph $G$ with vertex set $V(G_1)\cup V(G_2)$ and edge set $E(G_1)\cup E(G_2)$. Let $T_1$ and $T_2$ be the co-trees of $G_1$ and $G_2$ with root $u_1$ and $u_2$, respectively. The co-tree $T$ of $G$ is the disjoint union of $T_1$ and $T_2$ with an additional node $u$ and two additional edges $\{u,u_1\}$ and $\{u,u_2\}$. Node $u$ is the {\em root} of $T$ labelled by $\union$. Node $u_1$ and $u_2$ are {\em successor} nodes of $u$.

\item If $G_1$ and $G_2$ are two undirected co-graphs, then the {\em join} of $G_1$ and $G_2$, denoted by $G_1 \join G_2$, is an undirected co-graph with vertex set $V(G_1)\cup V(G_2)$ and edge set $E(G_1)\cup E(G_2)\cup \{\{u,v\}\ \vert\ u\in V(G_1), v\in V(G_2)\}$. Let $T_1$ and $T_2$ be the co-trees of $G_1$ and $G_2$ with root $u_1$ and $u_2$, respectively. The co-tree $T$ of $G$ is the disjoint union of $T_1$ and $T_2$ with an additional node $u$ and two additional edges $\{u,u_1\}$ and $\{u,u_2\}$. Node $u$ is the {\em root} of $T$ labelled by $\join$. Node $u_1$ and $u_2$ are {\em successor} nodes of $u$.

\end{itemize}
\end{definition}

Undirected co-graphs are precisely those graphs that do not contain an induced $P_4$, i.e. a path with four vertices, thus their diameter is at most 2.

%
% Definition 2
%
\begin{definition}
\cite{OP07} Let $G$ be an undirected graph.
\begin{enumerate}
\item A vertex $u \in V(G)$ is {\em maximally distant} to a vertex $v\in V(G)$ in $G$ if $\forall u'\in N(u): d_G(u',v) \leq d_G(u,v)$.
\item Two vertices $u,v \in V(G)$ are {\em mutually maximally distant} if $u$ is maximally distant to $v$ and $v$ is maximally distant to $u$.
\item The {\em strong resolving graph} of $G$, denoted by $\GSR$, is an undirected graph with the same vertex set as $G$ and an edge between two vertices $u,v\in V(G)$ if an only if $u$ and $v$ are mutually maximally distant.
\end{enumerate}
\end{definition}

%
% Theorem 1
%
\begin{theorem}\label{t:SMD=VC}
\cite{OP07} Let $G$ be an undirected graph, then $\SMD(G) = \tau(\GSR)$, where $\tau(\GSR)$ is the size of a minimum vertex cover of $\GSR$.
\end{theorem}

Instead of looking for a minimum vertex cover of $\GSR$, we can also look for a maximum independent set of $\GSR$ or equivalently, for a maximum clique of $\GSRC$. That is, we can calculate the strong metric dimension of $G$ by
$$\tau(\GSR) \ = \ n - \alpha(\GSR) \ = \ n - \omega(\GSRC),$$ where $n$ is the number of vertices, $\alpha(\GSR)$ is the size of a maximum independent set of $\GSR$ and $\omega(\GSRC)$ is the size of a maximum clique of $\GSRC$.

%
% Lemma 1
%
\begin{lemma}\label{t:TTadjGSR}
Let $G$ be an undirected graph.
\begin{enumerate}
\item
If two vertices $u,v \in V(G)$ of $G$ are true (false) twins in $G$, then $u$ and $v$ are false (true) twins in $\overline{G}$.
\item
If two vertices $u,v \in V(G)$ of $G$ are twins in $G$, then $u$ and $v$ are true twins in $\GSR$.
\item
If $\{u, v\}$ is an edge of $G$, then $\{u, v\}$ is an edge in $\GSR$ if and only if $u$ and $v$ are true twins in $G$.
\end{enumerate}
\end{lemma}

\begin{proof} ~
\begin{enumerate}
\item
Obviously.
\item Let $w \in V(G)\setminus \{u,v\}$. If $u$ and $v$ are twins, then vertex $w$ and vertex $u$ are mutually maximally distant if and only if vertex $w$ and vertex $v$ are mutually maximally distant. That is, $u$ and $v$ are twins in $\GSR$. Since twins are mutually maximally distant from each other, they are true twins in $\GSR$.
\item
Let $\{u, v\} \in E(G)$ be an edge of $G$.

$"\Rightarrow"$: From the definition of "mutually maximally distant" it follows that for every vertex $u' \in N(u)$, $d_G(u', v) \leq d_G(u, v)$ and for every vertex $v' \in N(v)$, $d_G(u,v') \leq d_G(u,v)$. Since $d_G(u,v) = 1$, every neighbour of $u$ is either $v$ or a neighbour of $v$ and every neighbour of $v$ is either $u$ or a neighbour of $u$. Therefore, $u$ and $v$ are true twins in $G$. \\

$"\Leftarrow"$: If $u$ and $v$ are true twins in $G$, every neighbour of $u$ is either $v$ or a neighbour of $v$ and every neighbour of $v$ is either $u$ or a neighbour of $u$. Therefore, for every vertex $u' \in N(u)$ $d_G(u', v) \leq d_G(u, v)$ and  for every vertex $v' \in N(v)$ $d_G(v', u) \leq d_G(v, u)$ and thus $u$ and $v$ are mutually maximally distant in $G$.
\end{enumerate}
\end{proof}

%
% Lemma 2
%
\begin{lemma}\label{t:GSRisGComp}
If $G$ is an undirected graph of diameter at most 2, then $$E(\GSR) = E(\overline{G}) \cup \{\{u,v\}\ \vert\ u,v \text{ are true twins in } G\}.$$
\end{lemma}

\begin{proof}
Since $d(G) \leq 2$, for every distinct pair of vertices $u, v\in V(G)$ either $d_G(u,v) = 1$ or $d_G(u,v) = 2$.
If two vertices $u, v\in V(G)$ are mutually maximally distant in $G$, that is if $\{u,v\} \in E(\GSR)$, then
either
\begin{enumerate}
\item
$d_G(u,v) = 2$ and $\{u,v\} \in E(\overline{G})$ or
\item $d_G(u,v) = 1$ and $\{u,v\} \not\in E(\overline{G})$ and by Lemma \ref{t:TTadjGSR}.3, $u$ and $v$ are true twins.
\end{enumerate}
If $\{u,v\} \in E(\overline{G})$ then $\{u,v\} \not\in E(G)$ and $d_G(u,v) = 2$ and $u,v$ are mutually maximally distant in $G$ and thus $\{u,v\} \in E(\GSR)$. 
\end{proof}

Since the complement $\overline{G}$ of a co-graph $G$ is a co-graph by the definition of co-graphs and since true twins of $G$ are false twins in $\overline{G}$ by Lemma \ref{t:TTadjGSR}.1 and since the insertion of edges between false twins with non-empty neighbourhood does not increase the distances between two vertices in a graph, it follows that $\GSR$ is a co-graph if $G$ is a co-graph.

%
% Corollary 1
%
\begin{corollary}
If $G$ is an undirected co-graph, then $\GSR$ is an undirected co-graph.
\end{corollary}

%
% Lemma 3
%
\begin{lemma}\label{t:removeTT}
Let $G$ be an undirected graph and $u,v \in V(G)$ be two true twins of $G$, then $\SMD(G) = \SMD(G\setminus v) + 1$.
\end{lemma}

\begin{proof}
Let $u$ and $v$ be two twins of $G$ and let $G' = G\setminus v$.
By Lemma \ref{t:TTadjGSR}.2 $u$ and $v$ are true twins in $\GSR$.
Two vertices $w_1,w_2 \in V(G) \setminus \{u,v\}$ are mutually maximally distant in $G$ if and only if they are mutually maximally distant in $G'$. Since $u$ and $v$ are true twins, for every vertex $w \in V(G) \setminus \{u,v\}$, the vertices $u$ and $w$ are mutually maximally distant in $G$ if and only of they are mutually maximally distant in $G'$. Note that this does not hold if $u$ and $v$ are false twins. This reasoning shows that $\GSRP = \GSR \setminus v$.

Every vertex cover $U'$ of $\GSRP$ can be converted into a vertex cover $U$ of $\GSR$ by adding vertex $v$ to $U'$, that is $\SMD(G')+1 \geq \SMD(G)$.
Every vertex cover $U$ of $\GSR$ contains vertex $u$ or vertex $v$ or both vertices, because $u$ and $v$ are adjacent in $\GSR$. If $U$ does not contain vertex $v$, then replace vertex $u$ in $U$ by vertex $v$. The resulting set $U$ is again a vertex cover of $\GSR$ and $U \setminus \{v\}$ is a vertex cover for $\GSRP$, which implies $\SMD(G)-1 \geq \SMD(G')$.
\end{proof}

Removing a twin in an undirected graph $G$ does not destroy the twin property of the remaining pairs of vertices.
If $u_1, u_2 \in V(G)$ are twins in $G$ and $v_1,v_2 \in V(G)$ are also twins in $G$, then $u_1, u_2$ remain twins in $G \setminus v_1$ and $G \setminus v_2$. This property is the prerequisite for the following theorem. 
 
%
% Theorem 2
%
\begin{theorem}
\label{T02}
Let $G$ be an undirected co-graph with $n$ vertices. Let $G'$ be the graph obtained by successively removing true twins from $G$, then $$\SMD(G) = n - \omega(G'),$$ where $\omega(G')$ is the size of a maximum clique in $G'$.
\end{theorem}

\begin{proof}
Let $n = \vert V(G) \vert$ and and $n' = |V(G')|$ then
$$\begin{array}{llll}
\SMD(G) & = & \SMD(G') + n - n' & \text{by repeatedly applying Lemma \ref{t:removeTT}} \\
        & = & \tau(\GSRP) + n - n' & \text{by Theorem \ref{t:SMD=VC}} \\
        & = & \tau(\overline{G'})  + n - n' & \text{by Lemma \ref{t:GSRisGComp}, because $G'$ has no true twins} \\
        & = & n' - \alpha(\overline{G'}) + n - n' & \text{because }\tau(\overline{G'}) = n' - \alpha(\overline{G'}) \\
        & = & n - \alpha(\overline{G'})& \\
        & = & n - \omega(G') &  \text{because }\alpha(\overline{G'}) = \omega(G'),
\end{array}
$$
where $\tau(G)$ and $\alpha(G)$ is the size of a minimum vertex cover and maximum independent set of a graph $G$, respectively.
\end{proof}

The strong metric dimension of an undirected co-graphs $G$ can be computed in linear time by computing the size of a maximum clique in an undirected co-graph $G$, in which all twins were removed beforehand. 
The computation of the size of a largest clique in $G$ can be carried out in linear time with a bottom-up processing of the co-tree $T$ of $G$.
A leaf gets the value $1$, an inner node $u$ labelled $\union$ gets the maximum of the values of the two successor nodes of $u$ and an inner node $u$ labelled by $\join$ gets the sum of the values of the two successor nodes of $u$.
The value of the root of $T$ is the size of a maximum clique in $G$.

The removal of true twins from an undirected co-graph can be carried out together with the bottom-up processing for the calculation of the size of a maximum clique. 
A simplest way to do this, is to use the so-called {\em canonical co-tree} in which successive union and successive join compositions are combined to one union and one join composition, respectively.
Then each path in $T$ from the root to a leaf is an alternating sequence of union and join compositions.
The canonical co-tree $T$ is not necessarily a binary tree but can also be computed in linear time \cite{HP05}.

It is easy to see, that two vertices in an undirected co-graph $G$ are true twins if and only if the two corresponding leaves in the canonical co-tree are successor nodes of the same join node. In order to compute the size of a maximum clique in $G$ without true twins, it is sufficient to count all successor nodes of a join node together as one unit only once. This is taken into account in the following Algorithm \ref{alg:undirected} using a variable $t$, which is set to a maximum of $1$. 

%
% Theorem 3
%
\begin{theorem}
\label{T03}
The strong metric dimension of undirected co-graphs can be computed in linear time.
\end{theorem}

%
% Algorithm 1
%
\begin{algorithm}
	\caption{Strong Metric Dimension of Undirected Co-graphs}
	\label{alg:undirected}
	\begin{algorithmic}[1]
		\Statex
		\Function{smd}{co-graph $G$}
			\State $T\gets$ canonical co-tree of $G$
			\State $n\gets \vert V(G)\vert$
			\State\Return $n-$\Call{max\_twinless\_clique}{$T$}
		\EndFunction
		\End
		\Statex
		\Function {max\_twinless\_clique}{canonical co-tree $T$}
			\State{$w\gets$ root of $T$}
			\If {$w$ is labelled by $\union$}
				\State{$m\gets 0$}
				\ForAll {subtrees $T'$ at root $w$}
					\State {$k=$\Call{max\_twinless\_clique}{$T'$}}
					\If {$k > m$}
						\State {$m\gets k$}
					\EndIf
				\EndFor
				\State {\Return $m$}
			\ElsIf {$w$ is labelled by $\join$}
				\State {$s\gets 0$}		
				\State {$t\gets 0$}		
				\ForAll {subtrees $T'$ at root $w$}
					\If {root of $T'$ is a leaf}
						\State{$t\gets 1$}
					\Else
						\State {$s\gets s+$\Call{max\_twinless\_clique}{$T'$}}
					\EndIf
				\EndFor
				\State {\Return $s+t$}
			\Else
				\State {\Return $1$}
			\EndIf
		\EndFunction
		\End
	\end{algorithmic}
\end{algorithm}

The above idea for the efficient computation of the strong metric dimension can of course also be used for all other graph classes that have a diameter of at most two, are closed with respect to the removal of twins and for which a maximum clique can be efficiently found. These would be, for example, perfect graphs with a diameter of at most two or graphs of bounded clique-width and a diameter of at most two. Here is an example of such a conclusion.

\begin{corollary}
The strong metric dimension of prefect graphs of diameter at most 2 can be computed in polynomial time.
\end{corollary}

\begin{proof}
Analogously to Theorem \ref{T02}, the  strong metric dimension of perfect graphs of diameter at most 2 can be computed by successively removing true twins from the graph by using Lemma \ref{t:removeTT} and calculating the size of a maximum clique in the remaining graph. Since removing true twins does not create holes or anti holes, the remaining graph is still perfect and of diameter at most 2. Therefore, Lemma \ref{t:GSRisGComp} still holds where a maximum clique of a prefect graph can be computed in polynomial time.
\end{proof}

%------------------------------------------------------------------------------
\section{The strong metric dimension of strongly connected directed co-graphs}
In this section we show how to determine the strong metric dimension of strongly connected directed co-graphs in linear time. Directed co-graphs are recursively defined as follows, see also \cite{BDGR97}:

%
% Definition 3
%
\begin{definition}[\bf Directed Co-graphs]
~
\begin{itemize}
\item A directed graph $G$ that consists of a single vertex $v$ is a directed co-graph. The co-tree $T$ of $G$ consists of a single node $u$ associated with vertex $v$ of $G$. Node $u$ is the {\em root} of $T$.

\item If $G_1$ and $G_2$ are two directed co-graphs, then the {\em union} of $G_1$ and $G_2$, denoted by $G_1\union G_2$, is a directed co-graph $G$ with vertex set $V(G_1)\cup V(G_2)$ and edge set $E(G_1)\cup E(G_2)$. Let $T_1$ and $T_2$ be the co-trees of $G_1$ and $G_2$ with root $u_1$ and $u_2$, respectively. The co-tree $T$ of $G$ is the disjoint union of $T_1$ and $T_2$ with an additional node $u$ and two additional edges $\{u,u_1\}$ and $\{u,u_2\}$. Node $u$ is the {\em root} of $T$ labelled by $\union$. Node $u_1$ is the {\em left successor node} of $u$ and node $u_2$ is the {\em right successor node} of $u$.

\item If $G_1$ and $G_2$ are two directed co-graphs, then the {\em join} of $G_1$ and $G_2$, denoted by $G_1 \join G_2$, is a directed co-graph with vertex set $V(G_1)\cup V(G_2)$ and edge set $E(G_1)\cup E(G_2)\cup \{(u,v),(v,u)\ \vert\ u\in V(G_1), v\in V(G_2)\}$. Let $T_1$ and $T_2$ be the co-trees of $G_1$ and $G_2$ with root $u_1$ and $u_2$, respectively. The co-tree $T$ of $G$ is the disjoint union of $T_1$ and $T_2$ with an additional node $u$ and two additional edges $\{u,u_1\}$ and $\{u,u_2\}$. Node $u$ is the {\em root} of $T$ labelled by $\join$. Node $u_1$ is the {\em left successor node} of $u$ and node $u_2$ is the {\em right successor node} of $u$.

\item If $G_1$ and $G_2$ are two directed co-graphs, then the {\em directed join} of $G_1$ and $G_2$, denoted by $G_1 \rjoin G_2$, is a directed co-graph with vertex set $V(G_1)\cup V(G_2)$ and edge set $E(G_1)\cup E(G_2)\cup \{(u,v)\ \vert\ u\in V(G_1), v\in V(G_2)\}$. Let $T_1$ and $T_2$ be the co-trees of $G_1$ and $G_2$ with root $u_1$ and $u_2$, respectively. The co-tree $T$ of $G$ is the disjoint union of $T_1$ and $T_2$ with an additional node $u$ and two additional edges $\{u,u_1\}$ and $\{u,u_2\}$. Node $u$ is the {\em root} of $T$ labelled by $\rjoin$. Node $u_1$ is the {\em left successor node} of $u$ and node $u_2$ is the {\em right successor node} of $u$.

\end{itemize}
\end{definition}

%
% Definition 4
%
\begin{definition}
\label{D04}
\cite{OP07} Let $G$ be a strongly connected directed graph and $u,v\in V(G)$.
\begin{itemize}
\item Vertex $u$ is maximally distant to vertex $v$, denoted by $u \MDT v$, if $\forall u'\in N^-(u): d(u',v)\leq d(u,v)$. 
\item Vertex $v$ is maximally distant from vertex $u$, denoted by $v \MDF u$, if $\forall v'\in N^+(v): d(u,v')\leq d(u,v)$. 
\item Vertex $u$ is mutually maximally distant to vertex $v$, denoted by $u \MMDT v$, if $u \MDT v$ and $v \MDF u$. 
%\item If $u$ is not mutually maximally distant to $v$, we denote this by $u \NMMDT v$.
\item The {\em strong resolving graph} of $G$, denoted by $\GSR$, is an undirected graph with the same vertex set as $G$ and an edge between two vertices $u,v\in V(G)$ if and only if $u$ is mutually maximally distant to vertex $v$ or $v$ is mutually maximally distant to vertex $u$, i.e., if $$(u \MMDT v) \vee (v \MMDT u).$$
\end{itemize}
\end{definition}

Note that $u \MMDT v$ does not imply $v \MMDT u$, see Figure \ref{fig:exampleCoTree}. 

%
% Figure 1
%
\begin{figure}
\centerline{\includegraphics[width=230pt]{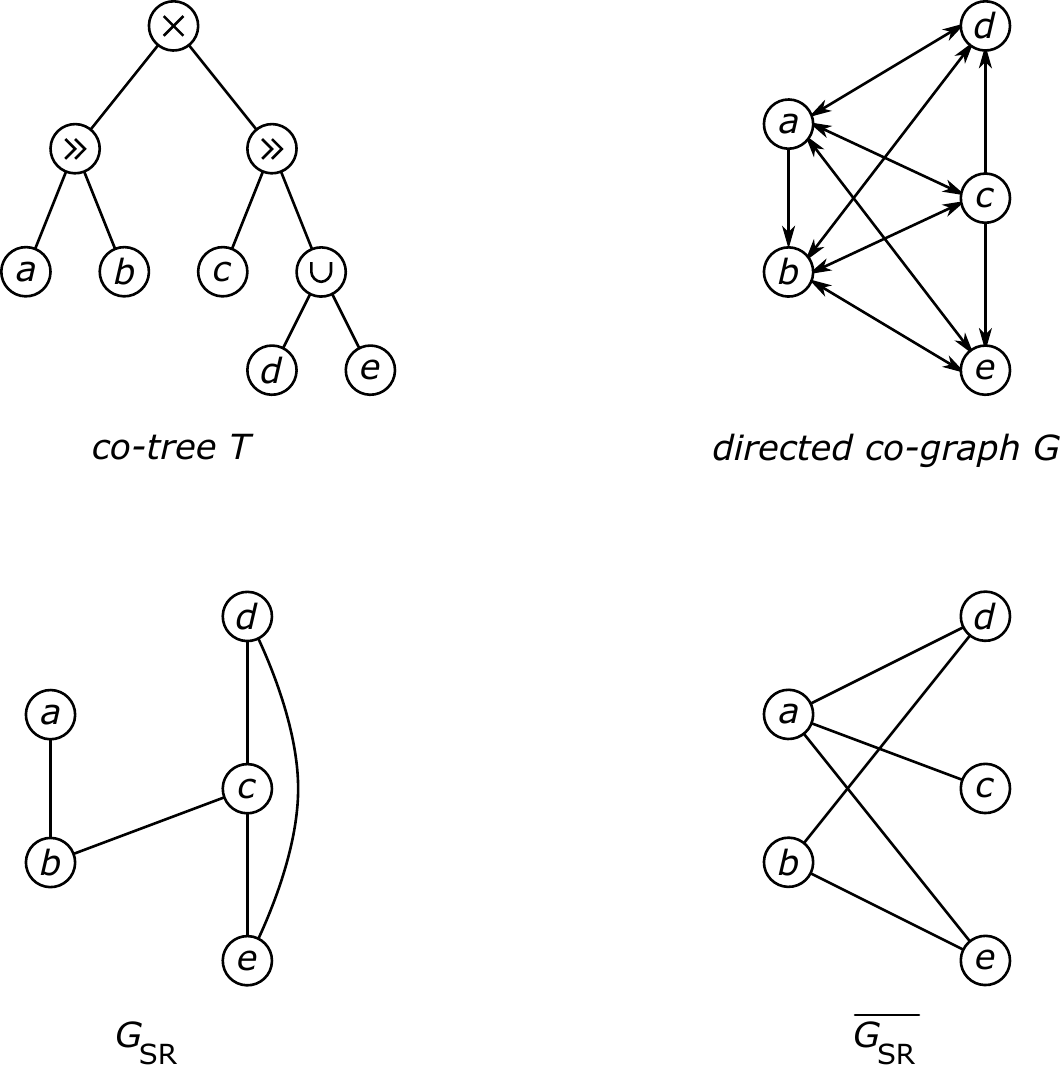}}

\caption{An example of a co-tree $T$ (top left) of a directed co-graph $G$ (top right), $\GSR$ (bottom left) and $\GSRC$ (bottom right). In this graph, the vertex $a$ is not maximally distant to vertex $c$, because $d$ and $e$ are in-neighbours of $a$, but not in-neighbours of $c$, thus $e\in N^-(a)$ and $d(e,c) > d(a,c)$. However, the vertex $c$ is maximally distant from vertex $a$, because every out-neighbour of $c$ is an out-neighbour of $a$ as well. In this example, we also have $c \MMDT b$, but not $b \MMDT c$. In this example, $\omega(\GSRC) = 2 = \alpha(\GSR)$ and $|V_G|-\alpha(\GSR) = \tau(\GSR) = 3 = \SMD(G)$.}
\label{fig:exampleCoTree}
\end{figure}

%
% Theorem 4
%
\begin{theorem}
\cite{OP07} Let $G$ be a strongly connected directed graph, then $\SMD(G) = \tau(\GSR)$.
\end{theorem}

We will now compute the strong metric dimension of directed co-graphs $G$ based on a similar idea as we used in Section \ref{sec:undirected}. Since we only consider strongly connected directed co-graphs, the last composition step is always a join operation if $G$ has more than one vertex. Instead of finding a vertex cover in $\GSR$, we determine a clique in $\GSRC$ as in Theorem \ref{T02}.

For the next two lemmas and the explanations of the algorithm for the computation of a strong resolving basis, let $G$ be a strongly connected directed co-graph and $T$ be the co-tree of $G$. Let $w$ be a node of $T$ and $T_w$ be the complete subtree of $T$ with root $w$. Let $T_l$ and $T_r$ be the left and right subtree of $T_w$ at root $w$. Let $G_w$, $G_l$ and $G_r$ be the subgraphs of $G$ induced by the vertices associated to the leafs of $T_w$, $T_l$ and $T_r$, respectively.

%
% Lemma 4
%
\begin{lemma}
\label{L04}
If node $w$ is labelled by $\union$ or $\rjoin$ then no vertex of $G_l$ is adjacent with a vertex of $G_r$ in $\GSRC$.
\end{lemma}

\begin{proof}
For every pair of vertices $u,v \in V(G)$ in a strongly connected directed co-graph $G$, we have $d_G(u,v) \leq 2$.
That is, if $d_G(u,v) = 2$, then $u \MMDT v$.
If $w$ is labelled by $\union$, then $d_G(u,v) = 2$ and $d_G(v,u) = 2$ and thus $u \MMDT v$ and $v \MMDT u$.
If $w$ is labelled by $\rjoin$, then obviously $d_G(v,u) = 2$ and $v \MMDT u$. However, the directed join composition which creates the edge $(u,v)$, also creates an edge from every in-neighbour of $u$ to $v$ and from $u$ to very out-neighbour of $v$. That is, in $G_w$ and thus also in $G$ there is no shortest path from an in-neighbour of $u$ via $u$ to $v$ and no shortest path from $u$ via $v$ to an out-neighbour of $v$, which implies $u \MMDT v$ in $G$.
\end{proof}

Lemma \ref{L04} shows that $\GSR$ contains all edges between the vertices of $G_l$ and the vertices of $G_r$ or equivalently, there is no edge between a vertex of $G_l$ and a vertex of $G_r$ in $\GSRC$.

Consider now the cases in which $w$ is labelled by $\join$. By definition, $\GSRC$ has an edge between $u$ and $v$ if and only if
$$\neg ((u \MMDT v) \vee (v \MMDT u))$$ or equivalently,
$$\neg (u \MMDT v) \wedge \neg (v \MMDT u)$$ or equivalently,
$$(\neg (u \MDT v) \vee \neg(v \MDF u)) \wedge (\neg (v \MDT u) \vee \neg (u \MDF v))$$ or equivalently by Definition \ref{D04},
\begin{enumerate}
\item
\begin{enumerate}
\item $\neg (\forall y\in N^-(u): d(y,v)\leq d(u,v))$ or
\item $\neg (\forall x\in N^+(v): d(u,x)\leq d(u,v))$
\end{enumerate}
and
\item
\begin{enumerate}
\item $\neg (\forall x\in N^-(v): d(x,u)\leq d(v,u))$ or
\item $\neg (\forall y\in N^+(u): d(v,y)\leq d(v,u))$.
\end{enumerate}
\end{enumerate}
or equivalently because we consider a join operation,
\begin{enumerate}
\item
\begin{enumerate}
\item $\exists y \in V(G)$ such that $(y,u)\in E(G)$ and $(y,v)\notin E(G)$ or
\item $\exists x \in V(G)$ such that $(v,x)\in E(G)$ and $(u,x)\notin E(G)$
\end{enumerate}
and
\item
\begin{enumerate}
\item $\exists x \in V(G)$ such that $(x,v)\in E(G)$ and $(x,u)\notin E(G)$ or
\item $\exists y \in V(G)$ such that  $(u,y)\in E(G)$ and $(v,y)\notin E(G)$.
\end{enumerate}
\end{enumerate}

Since we consider a join operation of $G_l$ and $G_r$, it follows that in Case 1.a and Case 2.b $y \in V(G_r)$ and in Case 1.b and Case 2.a $x \in V(G_l)$, see also Figure \ref{fig:F02}.

%
% Figure 2
%
\begin{figure}
\centerline{\includegraphics[width=\textwidth]{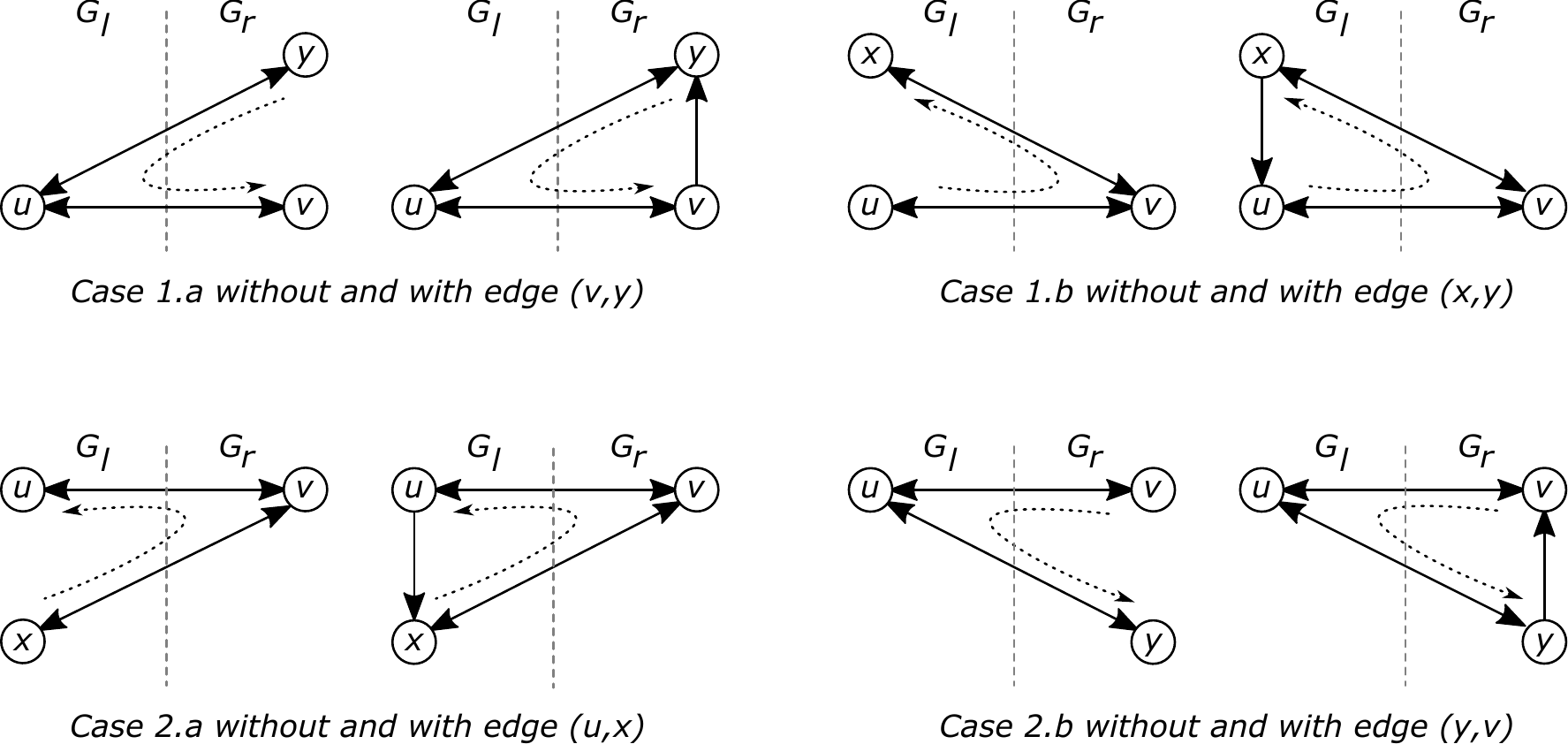}}

\caption{The four cases for the case that after a $\join$ operation between $G_l$ and $G_r$ there is an edge in $\GSRC$ between a vertex $u$ from $G_l$ and a vertex $v$ from $G_r$. The figures shows the edges in $G$. Since we consider a join operation, there is an edge from every vertex of $G_l$ to every vertex of $G_r$ in $G$ and vice versa.}
\label{fig:F02}
\end{figure}

We call a vertex $u \in V(G)$
\begin{itemize}
\item a {\em solitary vertex} of $G$ if there is a vertex $v \in V(G)$ such that $(u,v) \not\in E(G)$ and $(v,u) \not\in E(G)$,
\item an {\em in-vertex} in $G$ if there is a vertex $v \in V(G)$ such that $(v,u) \in E(G)$ and $(u,v) \not\in E(G)$,
\item an {\em out-vertex} in $G$ if there is a vertex $v \in V(G)$ such that $(u,v) \in E(G)$ and $(v,u) \not\in E(G)$,
\item an {\em in-out-vertex} in $G$ if $u$ is an in-vertex and an out-vertex in $G$.
\end{itemize}

A vertex $u$ can be a solitary, an in-vertex, an out-vertex and an in-out-vertex at the same time. 
If $u$ is an in-out-vertex, then it is also an in-vertex and an out-vertex. 

%
% Lemma 5
%
\begin{lemma}
\label{L05}
If node $w$ is labelled by $\join$, then two vertices $u \in V(G_l)$ and $v \in V(G_r)$ are adjacent in $\GSRC$ if and only if
\begin{enumerate}
\item $u$ is a solitary vertex in $G_l$ or
\item $v$ is a solitary vertex in $G_r$ or
\item $u$ is an in-out-vertex in $G_l$ or
\item $v$ is an in-out-vertex in $G_r$ or
\item $u$ is an in-vertex in $G_l$ and $v$ is an in-vertex in $G_r$ or
\item $u$ is an out-vertex in $G_l$ and $v$ is an out-vertex in $G_r$.
\end{enumerate}
\end{lemma}

\begin{proof}
$\Rightarrow$
\begin{enumerate}
\item If the conditions of Case 1.a and Case 2.a are satisfied, then either one of the two vertices $u$ and $v$ is a solitary vertex or $u$ and $v$ are both out-vertices in $G$.
\item If the conditions of Case 1.a and Case 2.b are satisfied, then either $v$ is a solitary vertex or an in-out-vertex in $G$.
\item If the conditions of Case 1.b and Case 2.a are satisfied, then either $u$ is a solitary vertex or an in-out-vertex in $G$.
\item If the conditions of Case 1.b and Case 2.b are satisfied, then either one of the two vertices $u$ and $v$ is a solitary vertex, or $u$ and $v$ are both in-vertices in $G$.
\end{enumerate}

$\Leftarrow$
\begin{enumerate}
\item
If $u$ is a solitary vertex in $G_l$, then the conditions of Case 1.b and Case 2.a are satisfied.
\item
If $v$ is a solitary vertex in $G_r$, then the conditions of Case 1.a and Case 2.b are satisfied.
\item
If $u$ is an in-out-vertex in $G_l$, then the conditions of Case 1.b and Case 2.a are satisfied.
\item
If $v$ is an in-out-vertex in $G_r$, then the conditions of Case 1.a and Case 2.b are satisfied.
\item
If $u$ is an out-vertex in $G_l$ and $v$ is an out-vertex in $G_r$, then the conditions of Case 1.a and Case 2.a are satisfied.
\item
If $u$ is an in-vertex in $G_l$ and $v$ is an in-vertex in $G_r$, then the conditions of Case 1.b and Case 2.b are satisfied.
\end{enumerate}
\end{proof}

A maximum clique in $\GSRC$ can now easily be computed with a bottom-up analysis of the co-tree $T$ for $G$. For this purpose, only four size values denoted by $m, s, i, o$ are required for each subtree $T_w$ of the co-tree $T$.
\begin{enumerate}
\item $m$ is the size of a maximum clique in $\GSRC|_{V(G_w)}$,
\item $s$ is the size of a maximum clique in $\GSRC|_{V(G_w)}$ with exclusively solitary or in-out vertices,
\item $i$ is the size of a maximum clique in $\GSRC|_{V(G_w)}$ with exclusively in-vertices and
\item $o$ is the size of a maximum clique in $\GSRC|_{V(G_w)}$ with exclusively out-vertices.
\end{enumerate}

The four size values $m,s,i,o$ for subtree $T_w$ can be easily computed from the four size values $m_l,s_l,i_l,o_l$ and the four size values $m_r,s_r,i_r,o_r$ for the left and right subtree $T_l$ and $T_r$ at root $w$, respectively.
In parallel to these computations, it is also possible to manage four vertex sets $M, S, O, I$ for maximum cliques in $G_w$ with any vertices, exclusively solitary or in-out vertices, exclusively in-vertices, and exclusively out-vertices, respectively. 

The size values $m, s, i, o$ for $G_w$ result from the maximum of several expressions formed from $m_l, s_l, i_l, o_l$ for $G_l$ and $m_r, s_r, i_r, o_r$ for $G_r$. Depending on which expression yields the maximum, the corresponding vertex sets $M, S, I, O$ for $G_w$ are assembled from the vertex sets $M_l, S_l, I_l, O_l$ for $G_l$ and the vertex sets $M_r, S_r, I_r, O_r$ for $ G_r $.

For a leaf node $w$ we have, $m=1, s=0, i=0, o=0$ and $M= \{\hat{w}\}, S= \emptyset, I= \emptyset, O= \emptyset$, where $\hat{w}$ is the vertex of $G$ associated with $w$.

If $w$ is labelled by $\union$, then
\begin{enumerate}
\item $m$ is the maximum of $m_l$ (let $M = M_l$) and $m_r$ (let $M = M_r$), 
\item $s$ is the maximum of $m_l$ (let $S = M_l$) and $m_r$ (let $S = M_r$),
\item $i$ is the maximum of $i_l$ (let $I = I_l$) and $i_r$ (let $I = I_r$) and
\item $o$ is the maximum of $o_l$ (let $O = O_l$) and $o_r$ (let $O = O_r$).
\end{enumerate}
This operation does not create any edges in $\GSRC$. All vertices become solitary vertices.

If $w$ is labelled by $\rjoin$, then
\begin{enumerate}
\item $m$ is the maximum of $m_l$ (let $M = M_l$) and $m_r$ (let $M = M_r$),
\item $s$ is the maximum of $s_l$ (let $S = S_l$), $s_r$ (let $S = S_r$), $i_r$ (let $S = I_r$) and $o_l$ (let $S = O_l$),
\item $i$ is the maximum of $m_r$ (let $I = M_r$) and $i_l$ (let $I = I_l$) and
\item $o$ is the maximum of $m_l$ (let $O = M_l$) and $o_r$ (let $O = O_r$).
\end{enumerate}
This operation also does not create any edges in $\GSRC$. The in-vertices of $G_l$ and the out-vertices of $G_r$ become in-out-vertices in $G_w$. All vertices of $G_l$ become out-vertices in $G_w$, all vertices of $G_r$ become in-vertices in $G_w$.

If $w$ is labelled by $\join$, then
\begin{enumerate}
\item $m$ is the maximum of $m_l+s_r$ (let $M=M_l \cup S_r$), $m_r+s_l$ (let $M=M_r \cup S_l$), $i_l+i_r$ (let $M=I_l \cup I_r$) and $o_l+o_r$ (let $M=O_l \cup O_r$),
\item $s= s_l+s_r$ (let $S=S_l \cup S_r$), 
\item $i= i_l+i_r$ (let $I=I_l \cup I_r$) and
\item $o= o_l+o_r$ (let $O=O_l \cup O_r$).
\end{enumerate}

This fully describes the bottom-up processing of the co-tree $T$ for the computation of a strong resolving basis for co-graph $G$. 

\begin{algorithm}
	\caption{Directed Strong Metric Dimension}
	\label{alg:directed}
	\begin{algorithmic}[1]
		\Function{dsmd\_di\_co-graph}{directed co-graph $G$}
			\State $T\gets$ binary co-tree of $G$
			\State $n\gets \vert V(G)\vert$
			\State $(m, s, i, o) \gets$ \Call{clique\_vector}{$T$}
			\State\Return $n - m$
		\EndFunction
		\End
		\Statex
		\Function{clique\_vector}{co-tree $T$}
			\State $w\gets$ root of $T$
			\If{$w$ is a leaf}
				\State\Return $(1,0,0,0)$
			\Else
				\State $T_l \gets$ left subtree of $w$
				\State $(m_l, s_l, i_l, o_l) \gets$ \Call{clique\_vector}{$T_l$}
				\State $T_r \gets$ right subtree of $w$
				\State $(m_r, s_r, i_r, o_r) \gets$ \Call{clique\_vector}{$T_r$}
				\If{$w$ is labelled $\union$}
					\State\Return $(\max\{m_l, m_r\}, \max\{m_l, m_r\}, \max\{i_l, i_r\}, \max\{o_l, o_r\})$
				\ElsIf{$w$ is labelled $\rjoin$}
					\State\Return $(\max\{m_l, m_r\}, \max\{s_l, s_r, i_r, o_l\}, \max\{m_r, i_l\}, \max\{m_l, o_r\})$
				\Else \Comment $w$ is labelled $\join$
					\State\Return $(\max\{m_l+s_r, m_r+s_l, i_l+i_r, o_l+o_r\}, s_l+s_r, i_l+i_r, o_l+o_r)$
				\EndIf
				\End
			\EndIf
			\End
		\EndFunction
		\End
	\end{algorithmic}
\end{algorithm}

Algorithms 2 shows such a bottom-up processing of the co-tree $T$, in that the size values $m,s,i,o$ are computed.
The algorithm can easily be extended to compute a minimum strong resolving set by keeping track of the vertex sets $M,S,I,O$. Since a binary co-tree of a directed co-graph can be found in linear time \cite{CP06}, we have proved the following theorem.

%
% Theorem 5
%
\begin{theorem}
A strong metric basis of a strongly connected directed co-graph is computable in linear time.
\end{theorem}

%
% Figure 3
%
\begin{figure}
\centerline{\includegraphics[width=380pt]{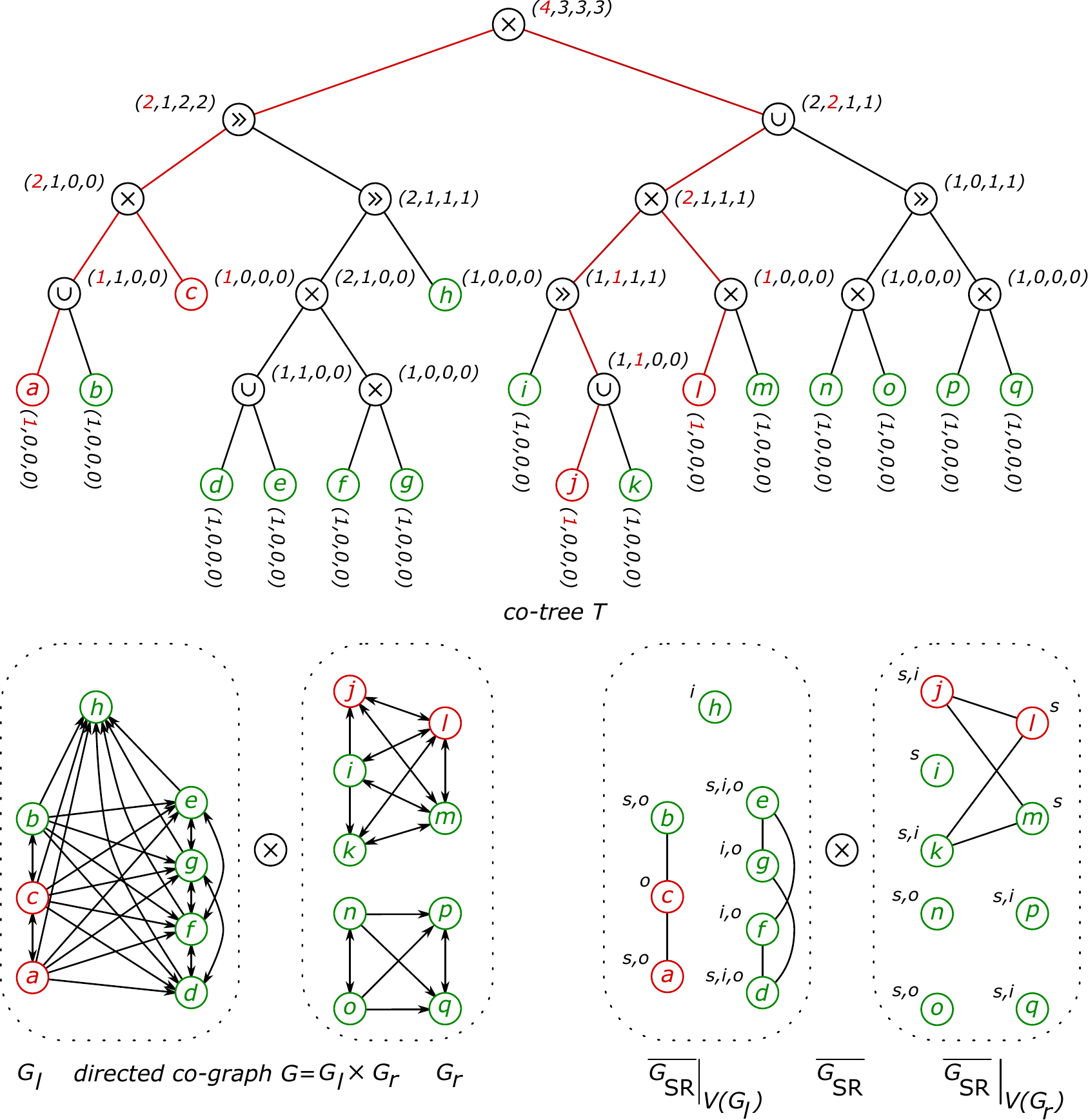}}

\caption{The figure shows a co-tree $T$ with the size values at the nodes calculated by Algorithm \ref{alg:directed}. The figure also shows bottom left the co-graph $G$ and bottom right the complement $\GSRC$ of the strong resolving graph $\GSR$.
The set of red vertices $\{a,c,j,l\}$ forms a maximum clique in $\GSRC$. There are several options to select a clique with four vertices in $\GSRC$. The set of green vertices $\{b,d,e,f,g,h,i,k,m,n,o,p,q\}$ is a strong resolving basis for the directed co-graph $G$. Let $G_l$ and $G_r$ be the graphs defined by the subtrees left and right at the root of $T$. Then $G = G_l \join G_r$. In $G_l$ the vertices $a,b,d,e$ are solitary vertices, $d,e,f,g,h$ are in-vertices and $a,b,c,d,e,f,g$ are out-vertices. In $G_r$, the vertices $i,j,k,l,m$ are solitary vertices, $i,j,p,q$ are in-vertices and $n,o$ are out-vertices. Graph $\GSRC$ is formed from the disjoint union of $\GSRC|_{V(G_l)}$ and $\GSRC|_{V(G_r)}$ and the edges between $u \in V(G_l)$ and $v \in V(G_r)$, where $u$ or $v$ is a solitary vertex or in-out-vertex, or $u$ and $v$ are both in-vertices, or both are out-vertices, see Lemma \ref{L05}. In the example all edges between vertices of $G_l$ and $G_r$ are inserted to get $\GSRC$. For the sake of clarity, these edges are not shown in the figure.
}
\label{fig:exampleAlgDirected}
\end{figure}

\section{Conclusions}
In this paper we have shown that {\sc Strong Metric Dimension} for undirected and directed co-graphs is decidable in linear time. We even presented linear time algorithms for the computation of minimum strong resolving sets for undirected and directed co-graphs.

The strong metric dimension as well as the directed strong metric dimension have not yet been extensively discussed in the literature. Developing efficient algorithms to compute the strong metric dimension for specific graph classes (both undirected and directional) is one of the most interesting challenges for us.

%%%%%%%%%%%%%%%%%%%%%%%%%%%%%%%%%%%%%%%%%%%%%%%%%%%%%%%%%%%%%%%%%%%%%%%%%%
\bibliography{Bibliography/bibliography_abbreviations.bib,Bibliography/bibliography_references.bib,Bibliography/bibliography_proceedings.bib}
%%%%%%%%%%%%%%%%%%%%%%%%%%%%%%%%%%%%%%%%%%%%%%%%%%%%%%%%%%%%%%%%%%%%%%%%%%

%%%%%%%%%%%%%%%%%%%%%%%%%%%%%%%%%%%%%%%%%%%%%%%%%%%%%%%%%%%%%%%%%%%%%%%%%%
\end{document}